\documentclass{article}
\usepackage{ijcai24}

\usepackage{times}
\usepackage{soul}
\usepackage{url}
\usepackage[hidelinks]{hyperref}
\usepackage[utf8]{inputenc}
\usepackage[small]{caption}
\usepackage{graphicx}
\usepackage{amsmath}
\usepackage{amsthm}
\usepackage{booktabs}
\usepackage[switch]{lineno}

\usepackage[ruled,vlined]{algorithm2e}

\usepackage{color}
\usepackage{subfigure}
\usepackage{multirow}
\usepackage{diagbox}

\usepackage{xspace}
\newcommand{\name}[0]{SCP\xspace}
\makeatletter
\DeclareRobustCommand\onedot{\futurelet\@let@token\@onedot}
\def\@onedot{\ifx\@let@token.\else.\null\fi\xspace}

\def\ie{\emph{i.e}\onedot}

\makeatother
\usepackage{xcolor}

\usepackage{amsthm}

\newtheorem{definition}{Definition}

\newtheorem{lemma}{Lemma}
\newtheorem{rul}{Rule}

\usepackage{makecell}

\usepackage{xspace}
\makeatletter
\DeclareRobustCommand\onedot{\futurelet\@let@token\@onedot}
\def\@onedot{\ifx\@let@token.\else.\null\fi\xspace}

\def\ie{\emph{i.e}\onedot}

\makeatother

\title{An Effective Branch-and-Bound Algorithm with New Bounding Methods \\
for the Maximum $s$-Bundle Problem}

\author{
    Jinghui Xue$^1$\thanks{The first two authors contribute equally.}\and 
    Jiongzhi Zheng$^{1*}$\and 
    Mingming Jin$^{1}$\and 
    Kun He$^1$\thanks{Corresponding author. Email: brooklet60@hust.edu.cn.}
    \affiliations
    $^1$School of Computer Science and Technology, Huazhong University of Science and Technology, China\\
}

\begin{document}

\maketitle

\begin{abstract}
The Maximum $s$-Bundle Problem (MBP) addresses the task of identifying a maximum $s$-bundle in a given graph. A graph $G=(V,E)$ is called an $s$-bundle if its vertex connectivity is at least $|V| - s$, where the vertex connectivity equals the minimum number of vertices whose deletion yields a disconnected or trivial graph.
MBP is NP-hard and holds relevance in numerous real-world scenarios emphasizing the vertex connectivity. Exact algorithms for MBP mainly follow the branch-and-bound (BnB) framework, whose performance heavily depends on the quality of the upper bound on the cardinality of a maximum $s$-bundle and the initial lower bound with graph reduction. 
In this work, we introduce a novel Partition-based Upper Bound (PUB) that leverages the graph partitioning technique to achieve a tighter upper bound compared to existing ones. 
To increase the lower bound, we propose to do short random walks on a clique to generate larger initial solutions. 
Then, we propose a new BnB algorithm that uses the initial lower bound and PUB in preprocessing for graph reduction, and uses PUB 
in the BnB search process for branch pruning. Extensive experiments 
with diverse $s$ values 
demonstrate the significant progress of our algorithm over state-of-the-art BnB MBP algorithms.
Moreover, our initial lower bound can also be generalized to other relaxation clique problems.
\end{abstract}

\section{Introduction}
The extraction of structured subgraphs within a graph is a critical task with numerous applications. One notable category among these is 
the clique model, which represents a thoroughly investigated subgraph structure. 
In an undirected graph, a clique constitutes a subset of vertices that induces a complete subgraph. 
However, in real-world scenarios like biological networks~\cite{k-defective-pro} and social networks~\cite{k-plex-pro}, 
dense subgraphs do not always need complete connectivity but allow some missing edges. Consequently, various types of relaxations and adaptations of the clique structure have been introduced, 
such as the quasi-clique~\cite{quasi-cliques-pro}, $s$-plex~\cite{k-plex-pro}, $s$-defective clique~\cite{k-defective-pro}, and $s$-bundle~\cite{pattillo2013clique}, where $s$ is usually a small integer, offering diverse solutions for practical applications across various domains.



A graph $G=(V,E)$ is termed an $s$-bundle if its vertex connectivity is at least $|V| - s$. Here, vertex connectivity refers to the minimum number of vertices whose removal results in a disconnected or trivial graph containing at most 1 vertex. 
The $s$-bundle structure frequently occurs in various network applications,
notably in scenarios emphasizing the vertex connectivity.
This is evident in the extraction of large and robustly connected subgroups in social networks~\cite{pattillo2013clique,veremyev2014finding}. However, recent prevailing of analyzing relaxation cliques focuses on the assessment of missing edges, such as the $s$-plex and $s$-defective clique models~\cite{k-plex-kplexS,WangZXK22,RGB,jiang2023Dise,wang2023fast,KDBB,zheng2023kd}, and there has been rather few exploration on $s$-bundle that highlights the vertex connectivity~\cite{RDS,zhou2022effective,hu2023listing}. Thus, we attempt to bridge the gap and 
address the Maximum $s$-Bundle Problem (MBP), which aims to find an $s$-bundle with the maximum number of vertices in a  graph.



MBP is NP-hard~\cite{RDS} and computationally challenging. Several exact algorithms based on the branch-and-bound (BnB) framework~\cite{LiQ10,McCreeshPT17} have been proposed for MBP. 
One is the generic Russian doll search (RDS) algorithm~\cite{RDS-proposed,RDS}, designed for diverse relaxed clique problems including MBP, and the other is a state-of-the-art algorithm called MSB~\cite{zhou2022effective}, which introduces novel upper bounds and reduction rules, and incorporates a multi-branching rule within the BnB framework.

BnB algorithms for the clique and relaxation clique problems usually maintain a progressively growing partial solution $S$ and 
the corresponding candidate set $C$ in each branching node. The algorithms calculate an upper bound $UB$ on the maximum size of feasible solutions containing $S$. The branching node can be pruned if $UB$ does not surpass the best solution found so far, which is referred to as the lower bound. 
Consequently, the effectiveness of the algorithm is significantly influenced by the quality of the upper bound as well as the initial lower bound.

MSB~\cite{zhou2022effective} introduces an upper bound by dividing the graph into independent sets, where no edge exists for any pairwise vertices. 
MSB claims that an independent set can only contribute at most $s$ vertices for an $s$-bundle, yet we found such a claim is conservative. Actually, considering the special property of $s$-bundle regarding the vertex connectivity, we find that a graph called $s$-component, where the size of its largest connected subgraph is no more than $s$, can only contribute at most $s$ vertices for an $s$-bundle. We relax the restriction of the independent set and allow each set to contain more vertices, resulting in a tighter upper bound.

Based on the above observation, we propose a new Partition-based Upper Bound (PUB) for MBP, which can help the BnB algorithm prune the branching nodes more efficiently. We further propose a new preprocessing method with a new reduction rule based on the PUB for reducing the input graph and a construction heuristic based on the short random walks~\cite{pearson1905problem} 
to calculate high-quality initial lower bound. 
Combining the above methods, we propose a new algorithm called $s$-Component Partition-based algorithm (\name) for MBP. 
Our main contributions are as follows:



\begin{itemize}
\item We propose a new Partition-based Upper Bound (PUB) for MBP, which can help obtain a tighter upper bound on the  $s$-bundle size. The PUB is used during the BnB searching process to prune the branching nodes and within the preprocessing for graph reduction.


\item We propose a new construction heuristic based on short random walks to have a high-quality initial lower bound 
than the general greedy construction heuristic based on vertex degrees. 
We also demonstrate its generality by applying it to $s$-defective clique problems. 


\item 
By employing our proposed upper bound and construction heuristic, we propose a new 
BnB algorithm for the MBP. 
Extensive experiments across various $s$ values and diverse datasets show the superiority of our algorithm over existing baselines. 
\end{itemize}


\section{Preliminaries}
\label{sec-Pre}

\subsection{Definitions}
\label{sec-def}
Let $G = (V, E)$ be an undirected graph, where $V$ is the set of vertices and $E$ the set of edges, $n=|V|$ and $m = |E|$ are the number of vertices and edges, respectively. The (vertex)  connectivity of graph $G$, denoted as $\kappa(G)$, is defined as the minimum number of vertices whose removal will result in a disconnected graph or a trivial graph with 
at most one vertex. 
For $S \subseteq V$, the induced subgraph $G[S]$ is called an $s$-bundle if $\kappa(G[S]) \geq |S| - s$ 
where $s$ is a positive integer. It is evident that a $1$-bundle is a clique. 

The set of neighbors of a vertex $v$ in $G$, \ie, adjacent vertices, is denoted as $N_G(v)$, and $|N_G(v)|$ indicates the degree of vertex $v$ in $G$. The common neighbor of two vertices $u$ and $v$ is defined as $N_G(u,v) = N_G(u) \cap N_G(v)$. We further define closed neighbors $N_G[v] = N_G(v) \cup \{v\}$ and $N_G[u,v] = N_G(u,v) \cup \{u,v\}$. 



To better describe our proposed 
upper bound, we define a new structure called $s$-component as follows.

\begin{definition}[$s$-component]
    An $s$-component is a graph where the largest connected subgraph contains at most $s$ vertices.
\end{definition}


\subsection{Prerequisites}

This subsection introduces several lemmas that are widely used and highly useful in solving MBP. 

\begin{lemma}
\label{property-1}
    If $S\subseteq V$ can induce an $s$-bundle in $G=(V,E)$, then for any $\mathcal{P} 
    \subseteq S$, 
    $\mathcal{P}$ can still induce an $s$-bundle.
\end{lemma}



 
Lemma $\ref{property-1}$ can be inferred by 
the hereditary property of graph theory~\cite{pattillo2013clique} and is useful for calculating the upper bound of the size of the maximum $s$-bundle in $G$. 


\begin{lemma}
\label{property-2}
    If $G=(V,E)$ is an $s$-bundle, 
    then $ |N_G(v)| \geq |V| - s $ for any vertex $v \in V$, 
    termed the degree-bound. 
\end{lemma}

Lemma \ref{property-2} can be proved easily by contradiction~\cite{zhou2022effective}.
It generally asserts that an $s$-bundle allows the absence of at most 
$s-1$ edges for each vertex. Otherwise, removing all neighbors of a vertex with more than $s-1$ non-neighbors can make the graph disconnected or trivial.



\subsection{Identifying an s-Bundle}
The determination of whether a graph $G$ qualifies as an $s$-bundle stands as a pivotal task in solving MBP. Here we introduce 
an existing method~\cite{RDS,zhou2022effective} for this purpose.

For an undirected graph $G = (V, E)$, to determine whether 
$\kappa(G) \geq n - s$, a directed flow graph $H = (U, A)$ needs to be constructed. 
For each vertex $u \in V$, two copies, denoted as $u'$ and $u''$, are created, and an arc $(u', u'')$ is established in $H$. Additionally, for every edge $(u, v) \in E$, two arcs $(u'', v')$ and $(v'', u')$ are generated in $H$. 
The capacity of each arc within $H$ is set to $1$. Consequently, the vertex connectivity $\kappa_G(u, v)$ of two non-adjacent vertices $u$ and $v$ in $G$ equals to the maximum $  (u'', v') $-flow within $H$, and we have $\kappa(G) = \min\limits_{\{u,v\} \in E} \kappa_G(u, v)$. If 
$\kappa(G) \geq |V| - s$ 
holds, then $V$ is an $s$-bundle. 
Note that $G$ is trivially an $s$-bundle for $s \geq |V|$, or $G$ is trivially not an $s$-bundle for $s < |V|$ if it is disconnected.

According to the above method, assessing whether an undirected graph $G = (V, E)$ is an $s$-bundle necessitates a time complexity of 
$O(m n^4)$ when we use the Dinic maxflow algorithm~\cite{dinlcalgorithm} with complexity of $O(m n^2)$ to determine each pairwise maximum flow. Additionally, determining whether a vertex can be added to the current $s$-bundle involves a time complexity of 
$O(m n^3)$. 
Despite that the time complexity is still in polynomial time, the duration required to ascertain whether a vertex set's induced subgraph is an $s$-bundle is substantial. Hence, the graph reduction method to minimize the scale of the problem and the upper bound method to prune the search tree size are crucial.


\section{Paritition-based Upper Bound}
This section introduces our proposed Partition-based Upper Bound (PUB). 
We first present the main idea and definition of PUB and then give an example for illustration. 
In the end, we introduce our proposed Partition-Bound algorithm for calculating the PUB. 

\subsection{Partition-based Bounding Rule}
We first briefly review the core upper bound used in the stat-of-the-art BnB algorithm of MSB~\cite{zhou2022effective}, which is based on graph coloring and denoted as the color-bound. Given a graph $G=(V,E)$, the color-bound divides $V$ into $k$ disjoint independent sets $I_1, \cdots, I_k$ by graph coloring techniques. Then,  each independent set $I_i$ ($1 \leq i \leq k$) can provide at most $\min\{|I_i|, s\}$ vertices for an $s$-bundle according to Lemma \ref{property-2}, and $\sum^k_{i=1} \min\{|I_i|, s\}$ is the color-bound of the size of the maximum $s$-bundle in $G$ based on Lemma~\ref{property-1}.




We observe that the color-bound 
still remains space for improvement. 
Although an independent set $I$ ($|I| > s$) can only contribute at most $s$ vertices, due to the relaxation property of $s$-bundle, the restriction of independent set can also be relaxed. Consequently, 
we consider the relaxation property of the problem and propose a partition-based new upper bound.


\begin{lemma}
\label{lemma-PUB}
    Given $s$, suppose graph $G=(V,E)$ can be partitioned into $k$ disjoint $s$-components, $\mathcal{P}^s_1=(V_1,E_1), \cdots, \mathcal{P}^s_k=(V_k,E_k)$. Then, each $s$-component $\mathcal{P}^s_i$ ($1 \leq i \leq k$) can provide at most $\min\{|V_i|,s\}$ vertices for an $s$-bundle.  Then 
    $PUB(G,s) = \sum^k_{i=1} \min\{|V_i|, s\}$ is an upper bound on the size of the maximum $s$-bundle in $G$. 
\end{lemma}

\begin{proof}
   For each $s$-component $\mathcal{P}^s_i$ ($1 \leq i \leq k$), if $|V_i| \leq s$, $V_i$ can provide at most $|V_i|$ vertices for an $s$-bundle. 
   If $|V_i| > s$, since the size of the largest connected subgraph in $\mathcal{P}^s_i$ is at most $s$, any subgraph $G'$ of $\mathcal{P}^s_i$ with more than $s$ vertices is disconnected, so we have $\kappa(G') = 0 < |G'| - s$, $G'$ is not an $s$-bundle, and the largest $s$-bundle in $\mathcal{P}^s_i$ contains at most $s$ vertices. Therefore, each $s$-component $\mathcal{P}^s_i$ can provide at most $\min\{|V_i|,s\}$ vertices for an $s$-bundle.
   Suppose $G[F]$ ($F \subseteq V$) is an $s$-bundle, then $F$ can be partitioned into $k$ disjoint set $F=\{F\cap V_1, ..., F\cap V_k\}$. According to Lemma~\ref{property-1}, $G[F\cap V_i]$ ($1 \leq i \leq k$) is also an $s$-bundle. Thus, we have $|F| = \sum_{i=1}^k{|F\cap V_i|} \geq \sum_{i=1}^k{\min\{|V_i|,s\}}$.
\end{proof}


Lemma \ref{lemma-PUB} presents the conception of our proposed PUB, which is somewhat tight because a graph that is not an $s$-component must contain a connected subgraph with more than $s$ vertices, which is an $s$-bundle with more than $s$ vertices. 
The PUB enables the expansion of a maximal independent set (\ie, not contained by any other independent set) into a maximal $s$-component. 
In other words, when compared to the color-bound approach, PUB requires no more sets to partition the graph. It is evident that PUB provides a more restrictive and efficient bound than the color-bound.
Moreover, an independent set can be regarded as a special 1-component, which fits well with the 1-bundle (\ie, clique), and our defined $s$-component is tailored for the $s$-bundle. 

\begin{figure}[t]
\centering
\subfigure[Partition of color-bound]{
\includegraphics[width= 0.4\columnwidth]{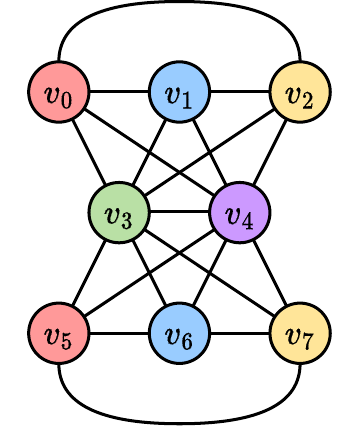}
\label{fig-color-bound}\hspace{1em}}
\subfigure[Partition of PUB]{
\includegraphics[width=0.4\columnwidth]{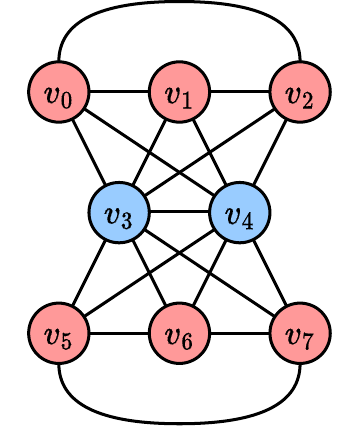}
\label{fig-partition-bound}}
\vspace{-1em}\caption{
Partitions for the maximum $3$-bundle problem.\vspace{-0.5em}}
\label{fig-bound}
\end{figure}

\begin{algorithm}[t]
\fontsize{10.1pt}{15}
\caption{Partition-Bound$(G,s)$}
\label{alg-parBound}
\LinesNumbered 
\KwIn{Graph $G=(V,E)$, positive integer $s$}
\KwOut{$PUB(G,s)$}
initialize the upper bound $UB \leftarrow 0$\;
\While {$V$ is not empty} {
    $I \leftarrow \emptyset$, $V'\leftarrow V$\;
    \While {$V'$ is not empty} {
        select a vertex $u$ in $V'$\;
        $V' \leftarrow V'\backslash N_G[u]$, $V \leftarrow V\backslash \{u\}$\;
        $I \leftarrow I \cup \{u\}$\;
    }
    \For{$u \in V$}{
        \If{$G[I \cup \{u\}]$ remains an $s$-component}{
            $V \leftarrow V\backslash \{u\}$, $I \leftarrow I \cup \{u\}$\;
        }
    }
    $UB \leftarrow UB+\min\{|I|, s\}$\;
}

\textbf{return} $UB$\;
\end{algorithm}

We provide an example to show how the color-bound and PUB are calculated. 
The task is to calculate an upper bound on the size of the maximum $3$-bundle in the graph $G$ shown in Figure \ref{fig-bound}. 
In Figure \ref{fig-color-bound}, the color-bound divides $G$ into 5 disjoint independent sets, $I_1 = \{v_0,v_5\}, I_2 = \{v_1,v_6\}, I_3 = \{v_2,v_7\}, I_4 = \{v_3\}$ and $I_5 = \{v_4\}$, the color-bound is $\sum_{i=1}^5{\min\{|I_i|,3\}} = 8$. 
However, the PUB can expand $G[I_1]$ and $G[I_4]$ to $3$-components.
Then, $G$ can be divided into 2 disjoint 3-components as Figure \ref{fig-partition-bound}, $\mathcal{P}^3_1 = (V_1,E_1)$ and $\mathcal{P}^3_2 = (V_2,E_2)$, where $V_1 = \{v_0,v_1,v_2,v_5,v_6,v_7\}$ and $V_2 = \{v_3,v_4\}$, and we have $PUB(G,3) = \sum_{i=1}^2{\min\{|V_i|,3\}} = 5$. 
In fact, the size of the maximum $3$-bundle in $G$ is exactly 5. The maximum $3$-bundles include the graph induced by $\{v_0,v_1,v_2,v_3,v_4\}$ and the graph induced by $\{v_0,v_1,v_3,v_4,v_6\}$.

\subsection{Calculation of the PUB}
We propose a Partition-Bound algorithm to divide the given graph into $s$-components and calculate the PUB, which is outlined in Algorithm \ref{alg-parBound}. 
The algorithm first initializes the upper bound $UB$ (line 1) and then iteratively partitions the input graph $G$ into $s$-components in a greedy manner to calculate the PUB (lines 2-11). 
For constructing each $s$-component, we first find a maximal independent set $I$ within the remaining vertex set $V$ (lines 3-7). 
Subsequently, for each remaining vertex $u \in V$, we identify if adding this vertex to $I$ still preserves that $G[I \cup \{u\}]$ is an $s$-component. If so, $u$ will be added to $I$ (lines 8-10). 
Finally, according to the PUB, we append $\min\{|I|, s\}$ to $UB$ (line $11$).
Note that we do not construct the $s$-component directly but try to enlarge a maximal independent set, which can usually lead to a larger $s$-component and make our PUB tighter than the color-bound.

Our approach employs the Disjoint Set Union method~\cite{tarjan1979class} to efficiently assess whether the insertion of a vertex preserves an $s$-component. 
By employing the method, the average time complexity for each query and merge operation is $O(\alpha(n))$, where $O(\alpha(n))$ represents the inverse of the Ackermann function, and can be considered as a constant time complexity~\cite{fredman1989cell}. Moreover, the time complexity of Algorithm \ref{alg-parBound} is dominated by the process of constructing the maximal $s$-component (lines 9-10), which needs to traverse each vertex $u \in V$ and query all its neighbors. Thus, its time complexity is $O(D|V|^2)$, where $D$ is the maximum degree of vertices in $G$.

\section{The Proposed \name Algorithm}




\subsection{Main Framework}

The main framework of our \name algorithm is outlined in Algorithm~\ref{alg-Framework}. Initially, \name calls the \texttt{GenerateLB} function to construct an initial $s$-bundle whose size is the initial lower bound $LB$ of 
maximum $s$-bundle (line 1). 
With the initial $LB$ value, the \texttt{Reduce} function is called to 
do graph reduction by removing vertices and edges that must not be in the maximum $s$-bundle (line 2), and the reduced graph is passed to the \texttt{BnB} function to identify the maximum $s$-bundle (line 3). During the BnB process, $LB$ will be updated whenever a larger $s$-bundle is discovered. Upon the completion of traversing the entire search tree by the \texttt{BnB} function, the final $LB$ value is returned as the size of the maximum $s$-bundle in $G$.


\subsection{Lower Bound Initialization}

The initial lower bound plays an important role in our BnB algorithm. A tighter lower bound can help reduce more vertices during preprocessing and prune more branching nodes in the BnB stage. The \texttt{GenerateLB} function in \name constructs a maximal $s$-bundle and regards its size as the initial lower bound. The algorithm first constructs a maximal clique in the input graph $G$ by iteratively adding a vertex with the largest degree and removing its non-neighbors from $G$. Then, we use the random walk technique~\cite{pearson1905problem}, which is widely used for detecting dense communities in social networks, to expand the maximal clique to a maximal $s$-bundle. 
Specifically, we iteratively employ a three-step lazy random work and expand the $s$-bundle by moving a node with the highest value from the candidate set. The lazy random work assigns unit weight to each vertex inside the current $s$-bundle and at each step it spreads half of the node weights to the neighbors.

The \texttt{GenerateLB} function with such lazy random walk can help the algorithm obtain a high-quality lower bound. 
Furthermore, experimental findings reveal that the random walk outperforms the greedy construction of the maximal $s$-bundle based on vertex degree. This superiority extends to a broader applicability, including generalizations for other relaxation clique problems.

\subsection{Preprocessing}
\label{sec-process}

Preprocessing is important in addressing massive sparse graphs. Given the initial lower bound $LB$, we can remove vertices and edges that must not be in any $s$-bundle with a size larger than $LB$. In this subsection, we first introduce some rules for identifying the removable vertices. The first two rules are based on the degree-bound introduced in Lemma~\ref{property-2}. 


\begin{algorithm}[t]
\fontsize{10.1pt}{15}
\caption{SCP$(G,s)$}
\label{alg-Framework}
\LinesNumbered 
\KwIn{Graph $G=(V,E)$, positive integer $s$}
\KwOut{Size of the maximum $s$-bundle in $G$}
 $LB \leftarrow$ GenerateLB$(G,s)$\;
$G\leftarrow$ Reduce$(G,s,LB)$\;
$LB\leftarrow$ BnB$(G,s,\emptyset,V,LB)$\; 
\textbf{return} $LB$\;
\end{algorithm}

\begin{rul}
\label{rule-1}
    Remove vertex $v$ from $G$ if it meets the condition $|N_G(v)|\leq LB - s$.
\end{rul}

\begin{rul}
\label{rule-2}
    Remove edge $(u,v)$ from $G$ if it meets the condition $|N_G(u,v)|\leq LB - 2s$.
\end{rul}

Rules~\ref{rule-1} and~\ref{rule-2} are straightforward since each vertex in an $s$-bundle has at most $s-1$ non-neighbors in the $s$-bundle. Based on our proposed PUB, we further propose a new rule as follows.

\begin{rul}
\label{rule-3}
    Remove vertex $v$ from $G$ if it meets the condition $PUB(G[N_G(v)],s)\leq LB - s$.
\end{rul}

\begin{proof}
    $PUB(G[N_G(v)],s)$ is an upper bound of the size of the maximum $s$-bundle in $G[N_G(v)]$, and an $s$-bundle containing $v$ can contain up to $s-1$ non-neighbors of $v$ according to Lemma~\ref{property-2}. Then, according to Lemma~\ref{property-1}, $PUB(G[N_G(v)],s) + s$ is an upper bound of the maximum $s$-bundle containing $v$.  
\end{proof}

\begin{algorithm}[t]
\fontsize{10.1pt}{15}
\caption{Reduce$(G,s,LB)$}
\label{alg-Preprocessing}
\LinesNumbered 
\KwIn{Graph $G=(V,E)$, positive integer $s$, lower bound $LB$}
\KwOut{Reduced graph $G$}
\While{true} {
$G'\leftarrow remove\_vertex\_with\_Rule1(G,s,LB)$ \;
$G'\leftarrow remove\_edge\_with\_Rule2(G',s,LB)$ \;
\lIf{$G'$ and $G$ is the same}{\textbf{break}}
\textbf{else} $G \leftarrow G'$\;
}
$G\leftarrow remove\_vertex\_with\_Rule3(G,s,LB)$ \;
\While{true} {
$G'\leftarrow remove\_edge\_with\_Rule2(G,s,LB)$ \;
$G'\leftarrow remove\_vertex\_with\_Rule1(G',s,LB)$ \;
\lIf{$G'$ and $G$ is the same}{\textbf{break}}
\textbf{else} $G \leftarrow G'$\;
}
\textbf{return} $G$\;
\end{algorithm}

Algorithm~\ref{alg-Preprocessing} depicts the procedure of the \texttt{Reduce} function that represents the preprocessing method. 
Since Rules~\ref{rule-1} and~\ref{rule-2} are computationally efficient, we first use them to quickly remove vertices that are simple to identify (lines 1-5) and then use Rule~\ref{rule-3} to focus on the ones that are difficult to identify (line 6). 
After reducing the graph by Rule~\ref{rule-3}, we further apply Rules~\ref{rule-2} and \ref{rule-1} alternatively until the graph cannot be further reduced. 
Actually, PUB can also be used for reducing each edge $(u,v)$ satisfying that $PUB(G[N_G(u,v)],s) \leq LB-2s$. 
We did not implement it since identifying the PUB $|E|$ times is heavily time-consuming.


\subsection{The Branch-and-Bound Process}
The BnB process in our \name algorithm is shown in Algorithm~\ref{alg-BnB}, where $S$ represents the set of vertices in the current growing $s$-bundle and $C$ is the corresponding candidate vertex set of $S$. Initially, the algorithm updates the best solution found so far (line 1) and applies Rule \ref{rule-1} together with the following rules derived from the degree-bound to reduce the candidate set (line 2). Actually, both Rule \ref{rule-5} and Rule \ref{rule-6} are removing vertices from the candidate set that cannot remain an $s$-bundle if added to $S$.


\begin{rul}
\label{rule-5}
    Remove vertex $u\in C$ from $C$ if it meets the condition 
    $|N_{G[S\cup \{u\}]}(u)| \leq |S| - s$.
\end{rul}

\begin{rul}
\label{rule-6}
    For each vertex $v\in S$ satisfying that $|N_{G[S]}(v)| = |S| - s$, remove vertex $u \in C$ that is not adjacent to $v$.
\end{rul}

Then, our proposed PUB is used to try to prune the current branching node (line 3), and the branching rules in MSB~\cite{zhou2022effective} are used in our algorithm to decide the subsequent branching nodes. 
The algorithm 
explores branches based on a special vertex $u_p$ that has the minimum degree in $G[S \cup C]$ (line 4).

When graph $G[S \cup C]$ is clearly not an $s$-bundle (lines 5-19), 
if $u_p \notin S$, the algorithm generates a branch that removes $u_p$ from $C$ (line 7) 
and 
then moves $u_p$ from $C$ to $S$ (line 8). 
Once $u_p$ is added to $S$, indicating that 
at most $t = s - 1 - |S \backslash N_{G'}[u_p]|$ of its non-neighbors can be added to $S$ (line 9). 
Thereafter, a multi-branching method with $t+1$ branches is used to branch on its non-neighbors $\{v_1,\cdots,v_c\}$ and prune some branches in advance, which executes in an incremental manner (lines 10-19). 
Specifically, the branches contains one branch that removes $v_1$ from $C$ (lines 13-14), $t-1$ branches that adds the first $i-1$ ($i \in \{2,\cdots,t\}$) non-neighbors of $u_p$ from $C$ to $S$ and further removes the $i$-th non-neighbor of $u_p$ from $C$ (lines 15-17), and one branch that moves the first $t$ non-neighbors of $u_p$ from $C$ to $S$ and further removes all its non-neighbors from $C$ (lines 18-19).

When graph $G[S \cup C]$ is possibly an $s$-bundle (lines 20-25), 
we first updates the LB if it is an $s$-bundle (lines 21-22).
Otherwise, we find a vertex $u\in C$ with the minimum degree in $G[S\cup C]$ (line 23), and use the binary branching method to branch on vertex $u$ (lines 24-25).

\begin{algorithm}[t]
\fontsize{10.1pt}{15}
\caption{BnB$(G,s,S,C,LB)$}
\label{alg-BnB}
\LinesNumbered 
\KwIn{Graph $G=(V,E)$, positive integer $s$, set of vertices in the current solution $S$, candidate set $C$, lower bound $LB$}
\KwOut{Size of the maximum $s$-bundle in $G$}
\lIf{$|S| > LB$}{$LB\leftarrow |S|$}
Reducing $C$ with Rules 1, 4, and 5\;
\lIf{$PUB(G[S\cup C], s) \leq LB$}{\textbf{return} $LB$}

$G' \leftarrow G[S\cup C], u_p \leftarrow \arg \min_{v\in S\cup C}|N_{G'}(v)|$\;
\eIf{$|N_{G'}(u_p)|< |S\cup C| - s$}{
\If{$u_p \notin S$}{
$LB \leftarrow$ BnB$(G, s, S, C\backslash \{u_p\}, LB)$\;
$S\leftarrow S\cup\{u_p\}, C\leftarrow C\backslash \{u_p\}$\;
}
$t \leftarrow s - 1 - |S \backslash N_{G'}[u_p]|$\;
$c \leftarrow |C \backslash N_{G'}(u_p)|$\;
$\{v_1, \cdots, v_{c}\} \leftarrow C \backslash N_{G'}(u_p)$\;
\For{$i \in \{1,\cdots,t+1\}$}{
\If{$i=1$}{$LB \leftarrow$ BnB$(G, s, S, C\backslash \{v_1\}, LB)$\;}
\ElseIf{$2\leq i \leq t$}{
$LB \leftarrow$ BnB$(G, s, S\cup \{v_1,\cdots,v_{i-1}\}$, 

$C\backslash \{v_1,\cdots,v_{i}\}, LB)$\;
}
\ElseIf{$i = t+1$}{
$LB \leftarrow$ BnB$(G, s, S\cup \{v_1,\cdots,v_{t}\}, C \cap N_{G'}(u_p)\}, LB)$\;
}
}
}
{
\If{$G'$is an $s$-bundle}{
$LB\leftarrow |S\cup C|$, \textbf{return} $LB$\;
}

$u \leftarrow \arg \min_{u\in C}|N_{G'}(u)|$\;
$LB \leftarrow$ BnB$(G, s, S, C\backslash \{u\}, LB)$\;
$LB \leftarrow$ BnB$(G, s, S\cup \{u\}, C\backslash \{u\}, LB)$\;

}
\textbf{return} $LB$\;
\end{algorithm}



\section{Experimental Results}
In this section, we first introduce the benchmarks and algorithms (also called solvers) used in experiments, then present and analyze the experimental results. All the algorithms are implemented in C++ and run on a server using an AMD EPYC 7742 CPU, running Ubuntu 18.04 Linux operation system. The cut-off time for each instance is 3,600 seconds, following the settings in the MSB algorithm~\cite{zhou2022effective}.

\begin{table*}[!t]
\centering
\footnotesize
 \scalebox{0.9}{
\begin{tabular}{l|ccc|ccc|ccc|ccc}\toprule
\multirow{2}{*}{~~$s$} & \multicolumn{3}{c|}{Facebook (114)} & \multicolumn{3}{c|}{RealWorld (102)} & \multicolumn{3}{c|}{DIMACS10 (82)} & \multicolumn{3}{c}{DIMACS2 (80)} \\
    & MSB & RDS & \name        & MSB          & RDS & \name        & MSB & RDS & \name         & MSB        & RDS         & \name        \\ \hline
~~2   & 15  & 32  & \textbf{~~64} & 35           & 27  & \textbf{~~58} & 14  & ~~6   & \textbf{~~28}  & 14         & \textbf{18} & 17          \\
~~3   & ~~6   & ~~2   & \textbf{~~43} & 34           & 15  & \textbf{~~52} & 14  & ~~3   & \textbf{~~22}  & ~~9          & \textbf{13} & ~~9           \\
~~4   & ~~4   & ~~0   & \textbf{~~32} & 32           & ~~7   & \textbf{~~48} & 14  & ~~2   & \textbf{~~22}  & ~~9          & ~~9           & \textbf{10} \\
~~5   & ~~4   & ~~0   & \textbf{~~26} & 31           & ~~4   & \textbf{~~47} & 14  & ~~1   & \textbf{~~24}  & ~~9          & ~~8           & \textbf{10} \\
~~6   & ~~4   & ~~0   & \textbf{~~25} & 31           & ~~4   & \textbf{~~44} & 13  & ~~1   & \textbf{~~22}  & 10         & ~~6           & \textbf{11} \\
~~8   & ~~6   & ~~0   & \textbf{~~21} & 30           & ~~2   & \textbf{~~43} & 12  & ~~1   & \textbf{~~21}  & 10         & ~~3           & \textbf{11} \\
10  & ~~5   & ~~0   & \textbf{~~22} & 29           & ~~1   & \textbf{~~42} & 13  & ~~1   & \textbf{~~20}  & ~~8          & ~~3           & \textbf{~~9}  \\
15  & ~~8   & ~~0   & \textbf{~~16} & 24           & ~~0   & \textbf{~~37} & ~~8   & ~~0   & \textbf{~~16}  & \textbf{~~7} & ~~3           & \textbf{~~7}  \\ \hline
Total & 52  & 34  & \textbf{249}         & 246 & 60  & \textbf{371}         & 102 & 15  & \textbf{175} & 76         & 63          & \textbf{84}
\\\bottomrule
\end{tabular}
 }
\vspace{-0.5em}\caption{Summary of the results of MSB, RDS, and \name on the number of instances  solved within a cut-off time of 3600s across four benchmarks. 
The best results appear in bold.\vspace{-0.5em}}
\label{tab-comp-results}
\end{table*}
\begin{table*}[]
\centering
\footnotesize
\resizebox{\linewidth}{!}{
\begin{tabular}{l|rr|rrrrrrrr|rrrrrrrr} \toprule
\multirow{3}{*}{Instance} &
  \multicolumn{1}{c}{\multirow{3}{*}{$|V|$}} &
  \multicolumn{1}{c|}{\multirow{3}{*}{$|E|$}} &
  \multicolumn{8}{c|}{$s=2$} &
  \multicolumn{8}{c}{$s=8$} \\ \cline{4-19} 
 &
  \multicolumn{1}{c}{} &
  \multicolumn{1}{c|}{} &
  \multicolumn{4}{c|}{\name} &
  \multicolumn{2}{c|}{MSB} &
  \multicolumn{2}{c|}{RDS} &
  \multicolumn{4}{c|}{\name} &
  \multicolumn{2}{c|}{MSB} &
  \multicolumn{2}{c}{RDS} \\
 &
  \multicolumn{1}{c}{} &
  \multicolumn{1}{c|}{} &
  \multicolumn{1}{c}{$|V'|$} &
  \multicolumn{1}{c}{$|E'|$} &
  \multicolumn{1}{c}{Tree} &
  \multicolumn{1}{c|}{Time} &
  \multicolumn{1}{c}{Tree} &
  \multicolumn{1}{c|}{Time} &
  \multicolumn{1}{c}{Tree} &
  \multicolumn{1}{c|}{Time} &
  \multicolumn{1}{c}{$|V'|$} &
  \multicolumn{1}{c}{$|E'|$} &
  \multicolumn{1}{c}{Tree} &
  \multicolumn{1}{c|}{Time} &
  \multicolumn{1}{c}{Tree} &
  \multicolumn{1}{c|}{Time} &
  \multicolumn{1}{c}{Tree} &
  \multicolumn{1}{c}{Time} \\ \hline
socfb-BU10 &
  19700 &
  637528 &
  5605 &
  121757 &
  \textbf{44.36} &
  \multicolumn{1}{r|}{\textbf{520.3}} &
  NA &
  \multicolumn{1}{r|}{NA} &
  NA &
  NA &
  5414 &
  137711 &
  \textbf{105.2} &
  \multicolumn{1}{r|}{\textbf{553.9}} &
  397.9 &
  \multicolumn{1}{r|}{3555} &
  NA &
  NA \\
socfb-Cal65 &
  11247 &
  351358 &
  1288 &
  33884 &
  \textbf{9.192} &
  \multicolumn{1}{r|}{\textbf{101.2}} &
  101.3 &
  \multicolumn{1}{r|}{868.8} &
  NA &
  NA &
  1841 &
  50983 &
  \textbf{33.00} &
  \multicolumn{1}{r|}{\textbf{116.6}} &
  214.8 &
  \multicolumn{1}{r|}{782} &
  NA &
  NA \\
socfb-UCF52 &
  14940 &
  428989 &
  5241 &
  185956 &
  \textbf{80.27} &
  \multicolumn{1}{r|}{\textbf{770.0}} &
  421.4 &
  \multicolumn{1}{r|}{2050} &
  NA &
  NA &
  5195 &
  203128 &
  \textbf{296.5} &
  \multicolumn{1}{r|}{\textbf{1029}} &
  500.6 &
  \multicolumn{1}{r|}{2063} &
  NA &
  NA \\
socfb-UCSB37 &
  14917 &
  482215 &
  810 &
  27649 &
  \textbf{6.518} &
  \multicolumn{1}{r|}{\textbf{102.7}} &
  280.4 &
  \multicolumn{1}{r|}{2817} &
  NA &
  NA &
  627 &
  21724 &
  \textbf{5.193} &
  \multicolumn{1}{r|}{\textbf{58.89}} &
  292.8 &
  \multicolumn{1}{r|}{1955} &
  NA &
  NA \\
socfb-wosn-friends &
  63731 &
  817090 &
  5425 &
  132570 &
  \textbf{483.1} &
  \multicolumn{1}{r|}{\textbf{1911}} &
  NA &
  \multicolumn{1}{r|}{NA} &
  NA &
  NA &
  3698 &
  98791 &
  \textbf{1439} &
  \multicolumn{1}{r|}{\textbf{1914}} &
  NA &
  \multicolumn{1}{r|}{NA} &
  NA &
  NA \\ \hline
ca-dblp-2010 &
  226413 &
  716460 &
  455 &
  13491 &
  \textbf{0.574} &
  \multicolumn{1}{r|}{\textbf{272.9}} &
  NA &
  \multicolumn{1}{r|}{NA} &
  NA &
  NA &
  844 &
  21716 &
  \textbf{2.327} &
  \multicolumn{1}{r|}{\textbf{327.6}} &
  NA &
  \multicolumn{1}{r|}{NA} &
  NA &
  NA \\
ca-dblp-2012 &
  317080 &
  1049866 &
  216 &
  11592 &
  \textbf{0.106} &
  \multicolumn{1}{r|}{\textbf{448.6}} &
  NA &
  \multicolumn{1}{r|}{NA} &
  NA &
  NA &
  216 &
  11592 &
  \textbf{0.160} &
  \multicolumn{1}{r|}{\textbf{396.0}} &
  NA &
  \multicolumn{1}{r|}{NA} &
  NA &
  NA \\
socfb-OR &
  63392 &
  816886 &
  5425 &
  132570 &
  \textbf{483.1} &
  \multicolumn{1}{r|}{\textbf{1663}} &
  NA &
  \multicolumn{1}{r|}{NA} &
  NA &
  NA &
  3698 &
  98791 &
  \textbf{1439} &
  \multicolumn{1}{r|}{\textbf{1339}} &
  NA &
  \multicolumn{1}{r|}{NA} &
  NA &
  NA \\
socfb-Penn94 &
  41536 &
  1362220 &
  6021 &
  162390 &
  \textbf{30.20} &
  \multicolumn{1}{r|}{\textbf{1528}} &
  NA &
  \multicolumn{1}{r|}{NA} &
  NA &
  NA &
  4054 &
  117357 &
  \textbf{57.35} &
  \multicolumn{1}{r|}{\textbf{771.1}} &
  NA &
  \multicolumn{1}{r|}{NA} &
  NA &
  NA \\
web-BerkStan &
  12305 &
  19500 &
  8816 &
  16011 &
  \textbf{41.98} &
  \multicolumn{1}{r|}{\textbf{8.958}} &
  59.52 &
  \multicolumn{1}{r|}{370.3} &
  NA &
  NA &
  8816 &
  16011 &
  \textbf{109.2} &
  \multicolumn{1}{r|}{\textbf{26.98}} &
  167.9 &
  \multicolumn{1}{r|}{460.9} &
  NA &
  NA \\ \hline
bio-pdb1HYS &
  36417 &
  36417 &
  189 &
  4341 &
  \textbf{0.322} &
  \multicolumn{1}{r|}{\textbf{0.750}} &
  8.889 &
  \multicolumn{1}{r|}{2.936} &
  NA &
  NA &
  60 &
  1632 &
  \textbf{10.76} &
  \multicolumn{1}{r|}{\textbf{2.699}} &
  498.4 &
  \multicolumn{1}{r|}{28.65} &
  NA &
  NA \\
chesapeake &
  39 &
  170 &
  12 &
  43 &
  \textbf{0.026} &
  \multicolumn{1}{r|}{\textbf{0.001}} &
  0.177 &
  \multicolumn{1}{r|}{\textbf{0.001}} &
  0.803 &
  0.004 &
  26 &
  119 &
  \textbf{0.087} &
  \multicolumn{1}{r|}{\textbf{0.002}} &
  0.293 &
  \multicolumn{1}{r|}{0.004} &
  14556 &
  151.1 \\
delaunay\_n12 &
  4096 &
  12264 &
  4096 &
  12264 &
  \textbf{16.51} &
  \multicolumn{1}{r|}{\textbf{4.610}} &
  17.6 &
  \multicolumn{1}{r|}{16.13} &
  5567 &
  580.2 &
  4046 &
  12114 &
  \textbf{24.89} &
  \multicolumn{1}{r|}{\textbf{3.782}} &
  27.08 &
  \multicolumn{1}{r|}{16.57} &
  NA &
  NA \\
fe-sphere &
  16386 &
  49152 &
  16386 &
  49152 &
  \textbf{23.85} &
  \multicolumn{1}{r|}{\textbf{11.46}} &
  179.6 &
  \multicolumn{1}{r|}{909.8} &
  NA &
  NA &
  16386 &
  49152 &
  \textbf{331.6} &
  \multicolumn{1}{r|}{\textbf{62.56}} &
  332.6 &
  \multicolumn{1}{r|}{883.1} &
  NA &
  NA \\
rgg\_n\_2\_17\_s0 &
  131070 &
  728753 &
  135 &
  936 &
  \textbf{0.219} &
  \multicolumn{1}{r|}{\textbf{22.51}} &
  NA &
  \multicolumn{1}{r|}{NA} &
  NA &
  NA &
  654 &
  4643 &
  \textbf{2.082} &
  \multicolumn{1}{r|}{\textbf{43.87}} &
  NA &
  \multicolumn{1}{r|}{NA} &
  NA &
  NA \\ \hline
c-fat200-1 &
  200 &
  1534 &
  90 &
  729 &
  \textbf{0.001} &
  \multicolumn{1}{r|}{\textbf{0.011}} &
  0.704 &
  \multicolumn{1}{r|}{0.019} &
  3.757 &
  0.031 &
  200 &
  1534 &
  \textbf{28.63} &
  \multicolumn{1}{r|}{\textbf{0.954}} &
  32.23 &
  \multicolumn{1}{r|}{0.976} &
  NA &
  NA \\
c-fat200-5 &
  200 &
  8473 &
  116 &
  4147 &
  \textbf{0.148} &
  \multicolumn{1}{r|}{0.147} &
  2.164 &
  \multicolumn{1}{r|}{\textbf{0.049}} &
  2.565 &
  0.083 &
  200 &
  8473 &
  \textbf{1720} &
  \multicolumn{1}{r|}{\textbf{75.78}} &
  1724 &
  \multicolumn{1}{r|}{81.74} &
  NA &
  NA \\
c-fat500-1 &
  500 &
  4459 &
  140 &
  1351 &
  \textbf{0.294} &
  \multicolumn{1}{r|}{\textbf{0.035}} &
  1.669 &
  \multicolumn{1}{r|}{0.117} &
  19.84 &
  0.301 &
  500 &
  4459 &
  \textbf{350.0} &
  \multicolumn{1}{r|}{\textbf{17.88}} &
  396.0 &
  \multicolumn{1}{r|}{18.55} &
  NA &
  NA \\
c-fat500-10 &
  500 &
  46627 &
  252 &
  19719 &
  \textbf{0.318} &
  \multicolumn{1}{r|}{1.961} &
  8.193 &
  \multicolumn{1}{r|}{\textbf{0.647}} &
  10.37 &
  1.471 &
  500 &
  46627 &
  \textbf{1783} &
  \multicolumn{1}{r|}{\textbf{332.9}} &
  1786 &
  \multicolumn{1}{r|}{373.9} &
  NA &
  NA \\
hamming6-4 &
  64 &
  704 &
  64 &
  480 &
  5.455 &
  \multicolumn{1}{r|}{0.041} &
  NA &
  \multicolumn{1}{r|}{NA} &
  \textbf{4.569} &
  \textbf{0.024} &
  64 &
  480 &
  \textbf{68.40} &
  \multicolumn{1}{r|}{\textbf{0.667}} &
  NA &
  \multicolumn{1}{r|}{NA} &
  NA &
  NA
\\\bottomrule
\end{tabular}
}
\vspace{-0.5em}\caption{Detailed results of MSB, RDS and \name on 40 representative MBP instances from four benchmarks with $s=2$ and $s=8$. For each benchmark of Facebook, RealWorld, DIMACS10, and DIMACS2, we select 5 graphs as ordered and grouped in the table. 
The search tree size is in $10^3$, and the time is in seconds. The best results appear in bold.\vspace{-1em}}
\label{tab-comp-instances3}
\end{table*}

\subsection{Benchmark Datasets}
We evaluate the algorithms on four public datasets that are widely used in studies related to the clique and various relaxation clique problems, including the Facebook\footnote{https://networkrepository.com/socfb.php} dataset that contains 114 massive sparse graphs derived from Facebook social networks, the Realword\footnote{http://lcs.ios.ac.cn/\%7Ecaisw/Resource/realworld\%20\\graphs.tar.gz} dataset that contains 102 massive sparse graphs from the Network Data Repository~\cite{RA15} (the extremely sparse and simple ``scc'' graphs are removed), the DIMACS10\footnote{https://www.cc.gatech.edu/dimacs10/downloads.shtml} dataset that contains 82 graphs with up to $1.05 \times 10^6$ vertices from the 10th DIMACS implementation challenge and is used in MSB~\cite{zhou2022effective}, and the DIMACS2\footnote{http://archive.dimacs.rutgers.edu/pub/challenge/graph/\\benchmarks/clique/} dataset that contains 80 dense graphs with up to 4,000 vertices from the 2nd DIMACS implementation challenge and is also used in MSB.

For each graph, we generate eight MBP instances with $s = 2,3,4,5,6,8,10,15$. Therefore, there are a total of $8 \times (114+102+82+80) = 3024$ MBP instances.

\subsection{Solvers}

To assess the performance of our proposed \name algorithm, we select the state-of-the-art BnB MBP algorithm called MSB~\cite{zhou2022effective} and the advanced and generic algorithm called RDS~\cite{RDS} as our baselines. To evaluate the effect of components in our \name algorithm, we also generate four variant algorithms. 
Details of all the algorithms in our experiments are summarized below.

\begin{itemize}
\item \textbf{MSB}: A BnB MBP algorithm with a coloring-based upper bound and multi-branching method\footnote{https://github.com/joey001/max-s-bundle\label{foot-MSB}}.
\item \textbf{RDS}: A basic BnB algorithm with a binary branching rule as its main framework for various relaxation clique problems, including MBP. 
It shows good performance on instances based on dense DIMACS2 graphs with small $s$ values.  The implemented version from \cite{zhou2022effective} is utilized\textsuperscript{\ref{foot-MSB}}.
\item \textbf{\name}: An implementation of our algorithm
\item \textbf{SCP$_\text{pre}^-$}: A variant of \name without preprocessing.
\item \textbf{SCP$_\text{randwalk}^-$}: A variant of \name that replaces the random walk method with the general greedy construction method, which prefers the vertex most connected to the current solution.
\item \textbf{SCP$_\text{color}^+$}: A variant of \name that replaces the PUB in both the preprocessing and the BnB stages with the color-bound in MSB.
\item \textbf{SCP$_\text{expand}^-$}: A variant of \name that does not expand independent sets into $s$-components. Instead, it starts with an empty $s$-component and sequentially traverses each vertex and try to add it to the $s$-component.
\end{itemize}

\subsection{Performance Comparison}

We first conduct a comprehensive comparison of \name, MBS, and RDS across all four benchmarks to assess their overall performance. 
The results are summarized in Table \ref{tab-comp-results}, providing the number of instances successfully solved by each algorithm within the given cut-off time. 
The instances are categorized based on the $s$ values for each benchmark.

The results reveal that \name significantly outperforms the baseline algorithms, particularly excelling on instances associated with sparse and large graphs. Notably, the RDS algorithm that lacks an upper bound and adheres to simple binary branching rules, experiences a substantial increase in the search tree size with growing $s$ values and graph sizes. It fails to solve most instances within the cut-off time when $s > 3$. Meanwhile, MSB faces challenges due to the lack of a tight enough upper bound and preprocessing method, and such limitation is particularly evident on sparse graphs.
With the benefits of our effective PUB, 
\name exhibits an excellent performance across various benchmarks and $s$ values. 
As a result, \name 
solves 9\%, 50\%, and 70\% more instances than MSB in the DIMACS2, RealWorld, and DIMACS10 benchmarks, respectively. Notably, \name successfully solves four times more instances than MSB in the Facebook benchmark, indicating a significant improvement.

We further provide detailed results for \name, MSB, and RDS in solving 20 representative instances from four benchmarks for $s = 2$ and $s = 8$, as outlined in Table~\ref{tab-comp-instances3}. The presented results include the number of vertices (column $|V|$) and edges (column $|E|$) for each original graph, along with the corresponding number of vertices (column $|V'|$) and edges (column $|E'|$) after the reduction by our preprocessing method. Additionally, the running time in seconds (column \textit{Time}) and the sizes of their entire search trees in $10^3$ (column \textit{Tree}) required to solve the instances are also provided. The symbol `NA' indicates that the algorithm cannot solve the instance within the given cut-off time.

The results show that when solving massive sparse graphs (with more than 40,000 vertices) the preprocessing method in PUB significantly reduces the graph size, enabling our \name to successfully solve them. In contrast, RSB and MSB fail to solve these instances within the given cut-off time.
For dense graphs, such as the DIMACS2 graphs c-fat200-1, c-fat500-1, and hamming6-4, preprocessing does not reduce any vertices or edges when $s=8$. Nevertheless, the BnB process based on PUB can still help \name solve these instances with much shorter running time as compared to the cut-off time. 
Consequently, both the search tree sizes and running time of \name are orders of magnitude smaller than those of MSB and RSB, demonstrating the superior performance of \name in solving instances with both small and large $s$ values.

\begin{figure}[!t]
\centering
\subfigure[
Ablation for preprocessing]{
\includegraphics[width=0.47\columnwidth]{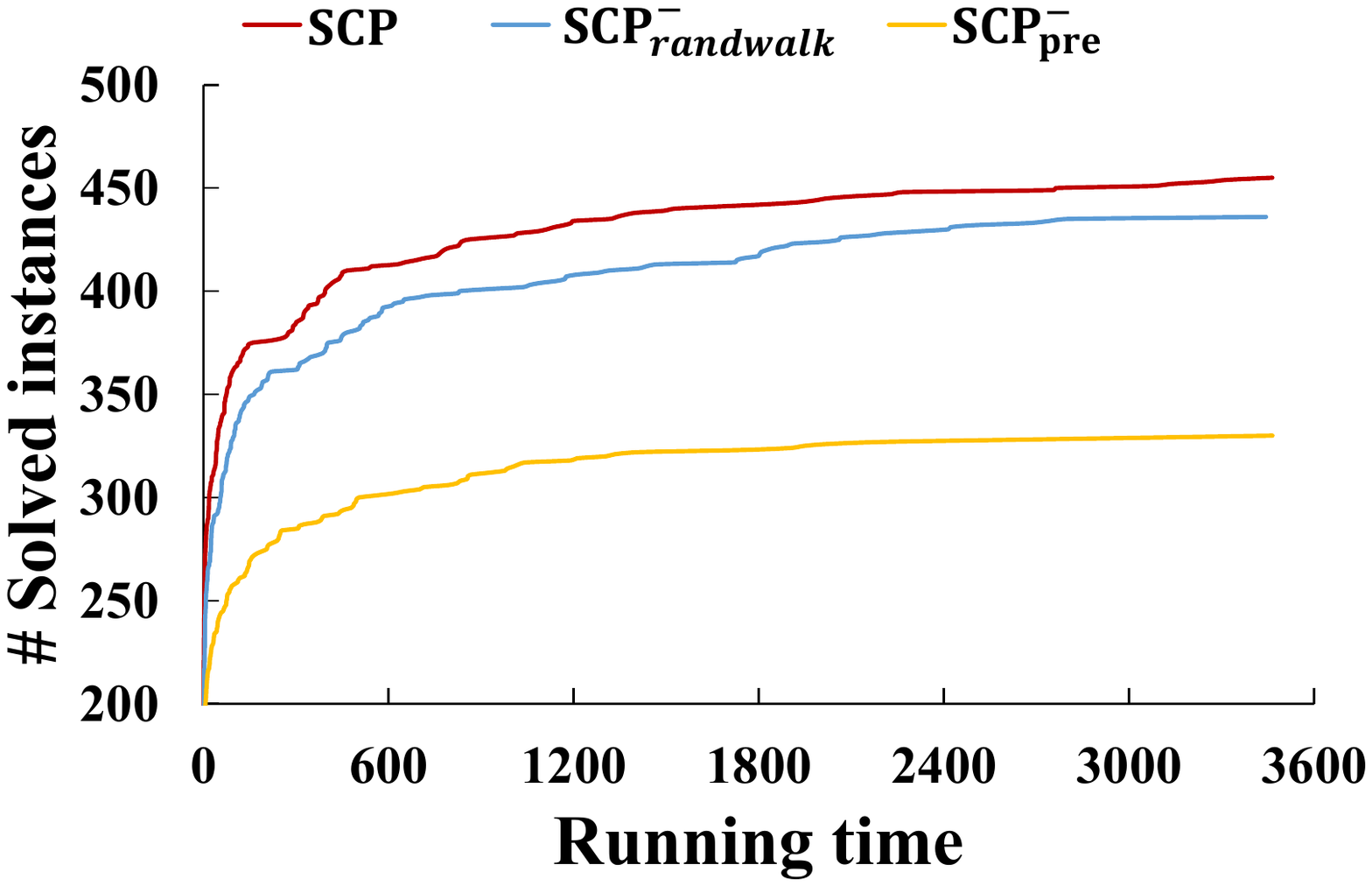}
\label{fig-Ablation1}}
\subfigure[
Ablation for bounds]{
\includegraphics[width=0.47\columnwidth]{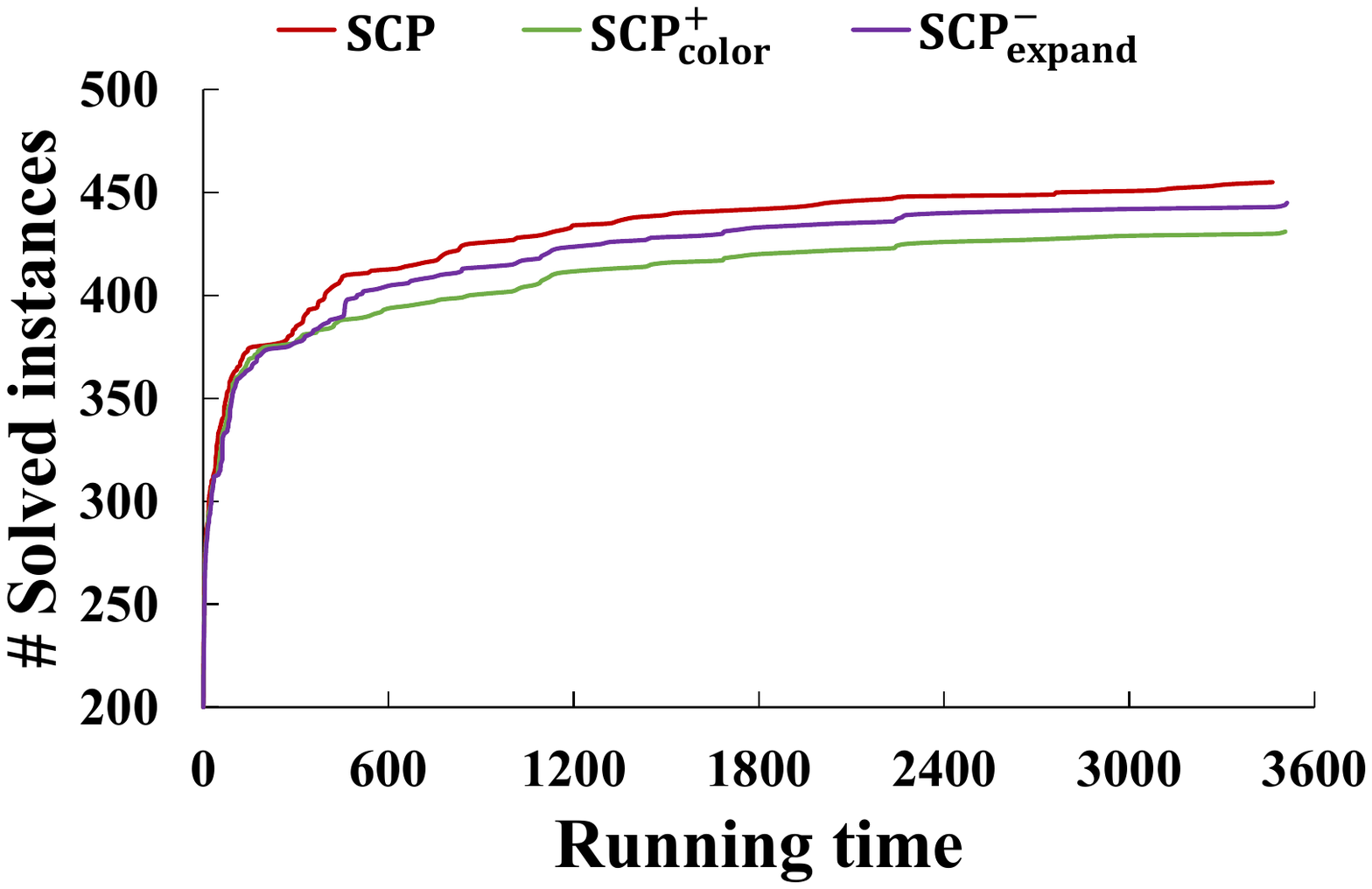}
\label{fig-Ablation2}}
\vspace{-1em}\caption{Ablation studies on the RealWorld and DIMACS2 benchmarks.\vspace{-0.5em}}
\label{fig-Ablation}
\end{figure}
\subsection{Ablation Study}

We perform ablation studies by comparing \name against its four variants grouped in two sets: SCP$_\text{randwalk}^-$ and SCP$_\text{pre}^-$, aiming to assess the effectiveness of the random walk method and the preprocessing method; SCP$_\text{color}^+$ and SCP$_\text{expand}^-$, aiming to assess the effectiveness of the PUB. 
We compare the algorithms on two typical benchmarks, DIMACS2 and RealWorld, which contain dense and massive sparse graphs, respectively, and we collect instances with all the eight $s$ values ($s = 2,3,...,15$ as in Table \ref{tab-comp-results}). 
The results are shown in Figure~\ref{fig-Ablation},
illustrating the evolution in the number of instances solved by each algorithm observed over the duration of running time (in seconds). 

In general, \name yields better performance than the four variants. 
Figure \ref{fig-Ablation1} clearly demonstrates the significant impact of preprocessing, and the employed construction heuristic based on random walk shows better performance than the traditional greedy construction method. 
Results in Figure \ref{fig-Ablation2} indicate that the PUB is more effective than the color-bound in reducing the graph and pruning the branches, and our method of expanding each independent set into a maximal $s$-component is more effective than finding the maximal $s$-component directly in obtaining tighter upper bounds.


\subsection{Generalization of our Lower Bound Method}
\label{sec-exp-LB}
\begin{table}[!t]
\centering
\footnotesize
\resizebox{\linewidth}{!}{
\begin{tabular}{c|c|cccccc}\toprule
Benchmark & \diagbox{Metric}{$s$} & 2    & 3    & 4    & 6    & 8    & 10   \\ \hline
DIMACS2    & $Num_{inc}$      & 67   & 67   & 69   & 71   & 73   & 74   \\
(80)                                & $LB_{inc}$ & 1.00 & 1.60 & 1.83 & 2.58 & 3.34 & 3.96 \\ \hline
RealWorld & $Num_{inc}$      & 52   & 53   & 55   & 59   & 59   & 60   \\
(102)                                & $LB_{inc}$ & 1.00 & 1.55 & 1.62 & 2.24 & 2.71 & 3.17
\\ \bottomrule
\end{tabular}
}
\vspace{-0.5em}\caption{Improvements over KDBB's lower bounds.\vspace{-0.5em}}
\label{tab-comp-results2}
\end{table}

To evaluate the generalization capability of our construction heuristic based on short random walks, we apply our method to generate the initial lower bound for the KDBB algorithm~\cite{KDBB}, a recent BnB algorithm for another practical relaxation clique problem, the Maximum $s$-Defective Clique Problem (MDCP). KDBB also uses a greedy construction heuristic guided by the vertex degree. We compare the methods in MDCP instances with 
various $s$ values on DIMACS2 and RealWorld benchmarks. The results are summarized in Table~\ref{tab-comp-results}, where $Num_{inc}$ indicates the number of instances that our short random walk method obtains a higher lower bound, and $LB_{inc}$ indicates the average increment on lower bounds.







The results show that our method demonstrates the ability to enhance the initial lower bounds in most instances.
Furthermore, with the increase of $s$, both the number of improved instances and the increments on lower bounds witness significant augmentation.
Importantly, for over 95\% of the tested instances, the running time of the random walk is less than 1 second, rendering it negligible in the BnB algorithm. We further leverage the improved lower bounds to enhance the efficiency of the KDBB algorithm, and the results are detailed in the Appendix.

\section{Conclusion}
This paper addressed the NP-hard Maximum $s$-Bundle Problem (MBP), 
a relaxation of the maximum clique problem. 
We proposed a Partition-based Upper Bound (PUB) by considering the relaxation property of $s$-bundle, 
and a heuristic based on short random walks to generate high-quality initial lower bound.
We thereby proposed a new branch and bound (BnB) algorithm that uses PUB in both preprocessing and BnB process, as well as the new intial lower bound. 
Comprehensive experiments employing various values of $s$ demonstrate that our algorithm consistently surpasses existing BnB algorithms. It exhibits superior performance and exceptional robustness across both dense and sparse graphs, resulting in significantly reduced search tree sizes and shorter overall running times.
We also apply the idea of our random walk-based initialization to another relaxation clique problem and demonstrate its generic performance. 
In our future work, we will consider fully utilizing the relaxation property of graphs to address various relaxation clique problems. 



\bibliographystyle{named}
\bibliography{ijcai24}

\appendix

\newpage
\twocolumn[
\begin{@twocolumnfalse}
\section*{\centering{\LARGE{Appendix of ``An Effective Branch-and-Bound Algorithm 
with New Bounding Methods  for the Maximum $s$-Bundle Problem''}}}
~\\
\end{@twocolumnfalse}
]
\begin{table*}[!ht]
\centering
\footnotesize
\scalebox{1}{
\resizebox{\linewidth}{!}{
\begin{tabular}{l|rr|rrrrrrrr|rrrrrrrr}\toprule
\multirow{3}{*}{Instance} &
  \multicolumn{1}{c}{\multirow{3}{*}{$|V|$}} &
  \multicolumn{1}{c|}{\multirow{3}{*}{$|E|$}} &
  \multicolumn{8}{c|}{$s=2$} &
  \multicolumn{8}{c}{$s=8$} \\ \cline{4-19} 
 &
  \multicolumn{1}{c}{} &
  \multicolumn{1}{c|}{} &
  \multicolumn{4}{c|}{KDBB} &
  \multicolumn{4}{c|}{KDBB$_{LB}^+$} &
  \multicolumn{4}{c|}{KDBB} &
  \multicolumn{4}{c}{KDBB$_{LB}^+$} \\
 &
  \multicolumn{1}{c}{} &
  \multicolumn{1}{c|}{} &
  \multicolumn{1}{r}{LB} &
  $|V'|$ &
  $|E'|$ &
  \multicolumn{1}{r|}{Time} &
  \multicolumn{1}{r}{LB} &
  $|V'|$ &
  $|E'|$ &
  Time &
  \multicolumn{1}{r}{LB} &
  $|V'|$ &
  $|E'|$ &
  \multicolumn{1}{r|}{Time} &
  \multicolumn{1}{r}{LB} &
  $|V'|$ &
  $|E'|$ &
  Time \\ \hline
bio-celegans &
  453 &
  2025 &
  9 &
  32 &
  203 &
  \multicolumn{1}{r|}{0.001} &
  \textbf{10} &
  \textbf{18} &
  \textbf{98} &
  \textbf{0.000} &
  9 &
  453 &
  2025 &
  \multicolumn{1}{r|}{0.406} &
  \textbf{11} &
  \textbf{424} &
  \textbf{1923} &
  \textbf{0.241} \\
bio-diseasome &
  516 &
  1188 &
  11 &
  21 &
  100 &
  \multicolumn{1}{r|}{0.000} &
  11 &
  21 &
  100 &
  0.000 &
  11 &
  282 &
  822 &
  \multicolumn{1}{r|}{0.03} &
  \textbf{12} &
  \textbf{222} &
  \textbf{689} &
  \textbf{0.014} \\
bio-dmela &
  7393 &
  25569 &
  7 &
  22 &
  66 &
  \multicolumn{1}{r|}{0.008} &
  \textbf{8} &
  \textbf{7} &
  \textbf{21} &
  \textbf{0.006} &
  7 &
  NA &
  NA &
  \multicolumn{1}{r|}{NA} &
  \textbf{9} &
  NA &
  NA &
  NA \\
bio-yeast &
  1458 &
  1948 &
  6 &
  12 &
  29 &
  \multicolumn{1}{r|}{0.000} &
  \textbf{7} &
  \textbf{0} &
  \textbf{0} &
  0.000 &
  6 &
  1458 &
  1948 &
  \multicolumn{1}{r|}{12.87} &
  \textbf{7} &
  1458 &
  1948 &
  \textbf{12.57} \\
ca-CondMat &
  17903 &
  196972 &
  26 &
  0 &
  0 &
  \multicolumn{1}{r|}{0.004} &
  26 &
  0 &
  0 &
  \textbf{0.002} &
  26 &
  84 &
  987 &
  \multicolumn{1}{r|}{0.006} &
  \textbf{27} &
  \textbf{51} &
  \textbf{619} &
  \textbf{0.002} \\
ca-CSphd &
  1882 &
  1740 &
  3 &
  1882 &
  1740 &
  \multicolumn{1}{r|}{343.7} &
  \textbf{4} &
  \textbf{15} &
  \textbf{20} &
  \textbf{0.000} &
  3 &
  1882 &
  1740 &
  \multicolumn{1}{r|}{55.57} &
  \textbf{5} &
  1882 &
  1740 &
  \textbf{54.11} \\
ca-Erdos992 &
  6100 &
  7515 &
  8 &
  11 &
  46 &
  \multicolumn{1}{r|}{0.001} &
  \textbf{9} &
  \textbf{0} &
  \textbf{0} &
  \textbf{0.000} &
  8 &
  6100 &
  7515 &
  \multicolumn{1}{r|}{876.6} &
  \textbf{10} &
  \textbf{921} &
  \textbf{2305} &
  \textbf{0.938} \\
ca-GrQc &
  4158 &
  13422 &
  44 &
  46 &
  1030 &
  \multicolumn{1}{r|}{0.003} &
  \textbf{45} &
  \textbf{0} &
  \textbf{0} &
  \textbf{0.000} &
  44 &
  46 &
  1030 &
  \multicolumn{1}{r|}{\textbf{0.002}} &
  \textbf{46} &
  0 &
  0 &
  0.003 \\
ca-netscience &
  379 &
  914 &
  9 &
  25 &
  92 &
  \multicolumn{1}{r|}{0.000} &
  9 &
  25 &
  92 &
  0.000 &
  9 &
  379 &
  914 &
  \multicolumn{1}{r|}{0.076} &
  \textbf{10} &
  \textbf{351} &
  \textbf{877} &
  \textbf{0.059} \\
ia-email-EU &
  32430 &
  54397 &
  12 &
  87 &
  1154 &
  \multicolumn{1}{r|}{0.027} &
  \textbf{13} &
  \textbf{32} &
  \textbf{339} &
  \textbf{0.019} &
  12 &
  1070 &
  11020 &
  \multicolumn{1}{r|}{5.427} &
  \textbf{15} &
  \textbf{331} &
  \textbf{4465} &
  \textbf{0.387} \\
ia-email-univ &
  1133 &
  5451 &
  12 &
  6 &
  15 &
  \multicolumn{1}{r|}{0.001} &
  12 &
  6 &
  15 &
  \textbf{0.000} &
  12 &
  341 &
  1562 &
  \multicolumn{1}{r|}{0.054} &
  \textbf{13} &
  \textbf{163} &
  \textbf{751} &
  \textbf{0.01} \\
ia-enron-large &
  33696 &
  180811 &
  18 &
  524 &
  14829 &
  \multicolumn{1}{r|}{2.164} &
  \textbf{19} &
  \textbf{425} &
  \textbf{11329} &
  \textbf{1.075} &
  18 &
  1736 &
  45994 &
  \multicolumn{1}{r|}{49.29} &
  \textbf{22} &
  \textbf{769} &
  \textbf{21837} &
  \textbf{11.11} \\
ia-enron-only &
  143 &
  623 &
  8 &
  18 &
  106 &
  \multicolumn{1}{r|}{0.001} &
  \textbf{9} &
  \textbf{16} &
  \textbf{90} &
  \textbf{0.000} &
  8 &
  143 &
  623 &
  \multicolumn{1}{r|}{0.07} &
  \textbf{11} &
  \textbf{98} &
  \textbf{458} &
  \textbf{0.007} \\
ia-fb-messages &
  1266 &
  6451 &
  5 &
  227 &
  967 &
  \multicolumn{1}{r|}{0.019} &
  \textbf{6} &
  \textbf{24} &
  \textbf{59} &
  \textbf{0.001} &
  5 &
  1266 &
  6451 &
  \multicolumn{1}{r|}{\textbf{2797}} &
  \textbf{7} &
  1266 &
  6451 &
  3149 \\
ia-infect-dublin &
  410 &
  2765 &
  15 &
  19 &
  162 &
  \multicolumn{1}{r|}{0.001} &
  \textbf{16} &
  \textbf{18} &
  \textbf{149} &
  0.001 &
  15 &
  168 &
  1276 &
  \multicolumn{1}{r|}{0.02} &
  \textbf{18} &
  \textbf{81} &
  \textbf{622} &
  \textbf{0.003} \\
ia-infect-hyper &
  113 &
  2196 &
  15 &
  77 &
  1501 &
  \multicolumn{1}{r|}{0.126} &
  \textbf{16} &
  \textbf{68} &
  \textbf{1290} &
  \textbf{0.080} &
  15 &
  108 &
  2117 &
  \multicolumn{1}{r|}{0.537} &
  \textbf{18} &
  \textbf{92} &
  \textbf{1844} &
  \textbf{0.486} \\
ia-reality &
  6809 &
  7680 &
  5 &
  57 &
  169 &
  \multicolumn{1}{r|}{0.001} &
  \textbf{6} &
  \textbf{9} &
  \textbf{26} &
  0.001 &
  5 &
  6809 &
  7680 &
  \multicolumn{1}{r|}{2107} &
  \textbf{7} &
  6809 &
  7680 &
  \textbf{2078} \\
ia-wiki-Talk &
  92117 &
  360767 &
  13 &
  958 &
  40110 &
  \multicolumn{1}{r|}{\textbf{820.6}} &
  \textbf{14} &
  \textbf{819} &
  \textbf{34337} &
  861.2 &
  13 &
  NA &
  NA &
  \multicolumn{1}{r|}{NA} &
  \textbf{16} &
  \textbf{1619} &
  \textbf{61292} &
  \textbf{3290} \\
inf-power &
  4941 &
  6594 &
  6 &
  25 &
  59 &
  \multicolumn{1}{r|}{0.001} &
  6 &
  25 &
  59 &
  0.001 &
  6 &
  4941 &
  6594 &
  \multicolumn{1}{r|}{745.5} &
  \textbf{8} &
  4941 &
  6594 &
  \textbf{722.8} \\
rec-amazon &
  91813 &
  125704 &
  5 &
  12643 &
  21020 &
  \multicolumn{1}{r|}{2307} &
  5 &
  12643 &
  21020 &
  \textbf{2264} &
  5 &
  NA &
  NA &
  \multicolumn{1}{r|}{NA} &
  \textbf{6} &
  NA &
  NA &
  NA \\
rt-retweet &
  96 &
  117 &
  4 &
  19 &
  27 &
  \multicolumn{1}{r|}{0.000} &
  \textbf{5} &
  \textbf{0} &
  \textbf{0} &
  0.000 &
  4 &
  96 &
  117 &
  \multicolumn{1}{r|}{0.015} &
  \textbf{6} &
  96 &
  117 &
  \textbf{0.014} \\
rt-twitter-copen &
  761 &
  1029 &
  4 &
  154 &
  289 &
  \multicolumn{1}{r|}{0.008} &
  \textbf{5} &
  \textbf{22} &
  \textbf{47} &
  \textbf{0.000} &
  4 &
  761 &
  1029 &
  \multicolumn{1}{r|}{1.993} &
  \textbf{6} &
  761 &
  1029 &
  \textbf{1.957} \\
sc-shipsec1 &
  140385 &
  1707759 &
  24 &
  240 &
  2760 &
  \multicolumn{1}{r|}{2.717} &
  24 &
  240 &
  2760 &
  \textbf{1.744} &
  24 &
  8020 &
  130505 &
  \multicolumn{1}{r|}{944.6} &
  \textbf{25} &
  \textbf{6274} &
  \textbf{89762} &
  \textbf{233.3} \\
soc-brightkite &
  56739 &
  212945 &
  36 &
  172 &
  5982 &
  \multicolumn{1}{r|}{1.494} &
  \textbf{37} &
  \textbf{165} &
  \textbf{5666} &
  \textbf{1.240} &
  36 &
  215 &
  8013 &
  \multicolumn{1}{r|}{2.282} &
  \textbf{40} &
  \textbf{178} &
  \textbf{6373} &
  \textbf{1.71} \\
soc-dolphins &
  62 &
  159 &
  5 &
  28 &
  65 &
  \multicolumn{1}{r|}{0.000} &
  \textbf{6} &
  \textbf{11} &
  \textbf{24} &
  0.000 &
  5 &
  62 &
  159 &
  \multicolumn{1}{r|}{0.014} &
  \textbf{7} &
  62 &
  159 &
  0.014 \\
soc-douban &
  154908 &
  327162 &
  11 &
  21 &
  161 &
  \multicolumn{1}{r|}{0.730} &
  \textbf{12} &
  \textbf{17} &
  \textbf{114} &
  \textbf{0.613} &
  11 &
  2801 &
  11064 &
  \multicolumn{1}{r|}{27.35} &
  \textbf{14} &
  \textbf{83} &
  \textbf{474} &
  \textbf{0.958} \\
soc-epinions &
  26588 &
  100120 &
  16 &
  193 &
  3301 &
  \multicolumn{1}{r|}{0.151} &
  \textbf{17} &
  \textbf{138} &
  \textbf{2242} &
  \textbf{0.109} &
  16 &
  617 &
  11005 &
  \multicolumn{1}{r|}{2.193} &
  \textbf{19} &
  \textbf{294} &
  \textbf{5843} &
  \textbf{1.038} \\
socfb-Berkeley13 &
  22900 &
  852419 &
  37 &
  848 &
  26786 &
  \multicolumn{1}{r|}{5.728} &
  \textbf{38} &
  \textbf{776} &
  \textbf{24537} &
  \textbf{4.540} &
  37 &
  1500 &
  46070 &
  \multicolumn{1}{r|}{22.51} &
  \textbf{41} &
  \textbf{1018} &
  \textbf{32385} &
  \textbf{9.341} \\
socfb-CMU &
  6621 &
  249959 &
  44 &
  56 &
  1492 &
  \multicolumn{1}{r|}{0.444} &
  \textbf{45} &
  56 &
  1492 &
  \textbf{0.397} &
  44 &
  297 &
  8363 &
  \multicolumn{1}{r|}{0.564} &
  \textbf{47} &
  \textbf{119} &
  \textbf{3398} &
  \textbf{0.449} \\
socfb-Duke14 &
  9885 &
  506437 &
  30 &
  1327 &
  50634 &
  \multicolumn{1}{r|}{\textbf{404.4}} &
  \textbf{31} &
  \textbf{1189} &
  \textbf{45732} &
  426.4 &
  30 &
  2698 &
  99887 &
  \multicolumn{1}{r|}{1302} &
  \textbf{33} &
  \textbf{1760} &
  \textbf{65836} &
  \textbf{1220} \\
socfb-Indiana &
  29732 &
  1305757 &
  44 &
  1540 &
  54558 &
  \multicolumn{1}{r|}{16.48} &
  \textbf{45} &
  \textbf{1225} &
  \textbf{43354} &
  \textbf{11.72} &
  44 &
  2566 &
  92410 &
  \multicolumn{1}{r|}{88.17} &
  \textbf{48} &
  \textbf{1860} &
  \textbf{65862} &
  \textbf{34.72} \\
socfb-MIT &
  6402 &
  251230 &
  29 &
  1033 &
  32448 &
  \multicolumn{1}{r|}{7.059} &
  \textbf{30} &
  \textbf{877} &
  \textbf{27732} &
  \textbf{5.168} &
  29 &
  1657 &
  60575 &
  \multicolumn{1}{r|}{43.96} &
  \textbf{33} &
  \textbf{1167} &
  \textbf{38982} &
  \textbf{17.38} \\
socfb-OR &
  63392 &
  816886 &
  25 &
  1132 &
  27730 &
  \multicolumn{1}{r|}{6.887} &
  \textbf{26} &
  \textbf{858} &
  \textbf{20832} &
  \textbf{4.683} &
  25 &
  3464 &
  85574 &
  \multicolumn{1}{r|}{166.3} &
  \textbf{29} &
  \textbf{1877} &
  \textbf{45300} &
  \textbf{32.73} \\
socfb-Penn94 &
  41536 &
  1362220 &
  38 &
  609 &
  17020 &
  \multicolumn{1}{r|}{3.861} &
  \textbf{39} &
  \textbf{487} &
  \textbf{13697} &
  \textbf{3.422} &
  38 &
  1889 &
  51910 &
  \multicolumn{1}{r|}{22.22} &
  \textbf{42} &
  \textbf{1106} &
  \textbf{29325} &
  \textbf{6.789} \\
socfb-Stanford3 &
  11586 &
  568309 &
  48 &
  151 &
  5288 &
  \multicolumn{1}{r|}{3.194} &
  \textbf{49} &
  \textbf{148} &
  \textbf{5125} &
  \textbf{3.181} &
  48 &
  562 &
  21246 &
  \multicolumn{1}{r|}{6.657} &
  \textbf{52} &
  \textbf{223} &
  \textbf{7577} &
  \textbf{3.852} \\
socfb-Texas84 &
  36364 &
  1590651 &
  48 &
  897 &
  39618 &
  \multicolumn{1}{r|}{72.82} &
  \textbf{49} &
  \textbf{799} &
  \textbf{35691} &
  \textbf{68.15} &
  48 &
  1416 &
  61781 &
  \multicolumn{1}{r|}{328.2} &
  \textbf{51} &
  \textbf{1192} &
  \textbf{51695} &
  \textbf{291.1} \\
socfb-UCLA &
  20453 &
  747604 &
  51 &
  136 &
  4367 &
  \multicolumn{1}{r|}{1.448} &
  \textbf{52} &
  \textbf{130} &
  \textbf{4043} &
  \textbf{1.233} &
  51 &
  303 &
  10178 &
  \multicolumn{1}{r|}{1.809} &
  \textbf{55} &
  \textbf{143} &
  \textbf{4743} &
  \textbf{1.467} \\
socfb-UConn &
  17206 &
  604867 &
  48 &
  211 &
  6883 &
  \multicolumn{1}{r|}{0.973} &
  \textbf{49} &
  \textbf{144} &
  \textbf{4818} &
  \textbf{0.775} &
  48 &
  301 &
  10108 &
  \multicolumn{1}{r|}{1.454} &
  \textbf{52} &
  \textbf{271} &
  \textbf{8589} &
  \textbf{1.059} \\
socfb-UCSB37 &
  14917 &
  482215 &
  51 &
  290 &
  10008 &
  \multicolumn{1}{r|}{1.756} &
  \textbf{52} &
  \textbf{170} &
  \textbf{6551} &
  \textbf{1.528} &
  51 &
  532 &
  18362 &
  \multicolumn{1}{r|}{3.668} &
  \textbf{55} &
  \textbf{320} &
  \textbf{11603} &
  \textbf{2.593} \\
socfb-UF &
  35111 &
  1465654 &
  51 &
  1568 &
  73133 &
  \multicolumn{1}{r|}{77.78} &
  \textbf{52} &
  \textbf{1526} &
  \textbf{70410} &
  \textbf{65.60} &
  51 &
  2010 &
  93977 &
  \multicolumn{1}{r|}{121.7} &
  \textbf{55} &
  \textbf{1717} &
  \textbf{80359} &
  \textbf{81.98} \\
socfb-UIllinois &
  30795 &
  1264421 &
  54 &
  565 &
  23799 &
  \multicolumn{1}{r|}{12.18} &
  \textbf{55} &
  \textbf{406} &
  \textbf{17969} &
  \textbf{10.58} &
  54 &
  676 &
  30508 &
  \multicolumn{1}{r|}{18.92} &
  \textbf{58} &
  \textbf{602} &
  \textbf{26096} &
  \textbf{12.67} \\
socfb-Wisconsin87 &
  23831 &
  835946 &
  33 &
  875 &
  27116 &
  \multicolumn{1}{r|}{6.475} &
  \textbf{34} &
  \textbf{778} &
  \textbf{23966} &
  \textbf{5.216} &
  33 &
  1863 &
  57527 &
  \multicolumn{1}{r|}{38.56} &
  \textbf{37} &
  \textbf{1057} &
  \textbf{33002} &
  \textbf{11.64} \\
soc-gowalla &
  196591 &
  950327 &
  29 &
  32 &
  489 &
  \multicolumn{1}{r|}{2.132} &
  \textbf{30} &
  \textbf{30} &
  \textbf{434} &
  \textbf{1.706} &
  29 &
  368 &
  10656 &
  \multicolumn{1}{r|}{3.306} &
  \textbf{32} &
  \textbf{103} &
  \textbf{2189} &
  \textbf{2.283} \\
soc-karate &
  34 &
  78 &
  5 &
  12 &
  25 &
  \multicolumn{1}{r|}{0.000} &
  \textbf{6} &
  \textbf{6} &
  \textbf{14} &
  0.000 &
  5 &
  34 &
  78 &
  \multicolumn{1}{r|}{0.003} &
  \textbf{7} &
  34 &
  78 &
  \textbf{0.002} \\
soc-slashdot &
  70068 &
  358647 &
  23 &
  201 &
  7947 &
  \multicolumn{1}{r|}{30.08} &
  \textbf{24} &
  \textbf{189} &
  \textbf{7506} &
  \textbf{28.60} &
  23 &
  306 &
  10976 &
  \multicolumn{1}{r|}{43.18} &
  \textbf{27} &
  \textbf{231} &
  \textbf{9068} &
  \textbf{37.07} \\
soc-wiki-Vote &
  889 &
  2914 &
  6 &
  93 &
  444 &
  \multicolumn{1}{r|}{0.005} &
  \textbf{7} &
  \textbf{27} &
  \textbf{132} &
  \textbf{0.001} &
  6 &
  889 &
  2914 &
  \multicolumn{1}{r|}{28.11} &
  \textbf{9} &
  889 &
  2914 &
  \textbf{22.89} \\
tech-as-caida2007 &
  26475 &
  53381 &
  16 &
  30 &
  359 &
  \multicolumn{1}{r|}{0.009} &
  \textbf{17} &
  \textbf{27} &
  \textbf{304} &
  \textbf{0.005} &
  16 &
  137 &
  2063 &
  \multicolumn{1}{r|}{0.129} &
  \textbf{18} &
  \textbf{78} &
  \textbf{1251} &
  \textbf{0.104} \\
tech-internet-as &
  40164 &
  85123 &
  16 &
  22 &
  219 &
  \multicolumn{1}{r|}{0.018} &
  \textbf{17} &
  22 &
  219 &
  \textbf{0.009} &
  16 &
  232 &
  3258 &
  \multicolumn{1}{r|}{0.176} &
  \textbf{19} &
  \textbf{71} &
  \textbf{1150} &
  \textbf{0.053} \\
tech-p2p-gnutella &
  62561 &
  147878 &
  4 &
  4343 &
  5388 &
  \multicolumn{1}{r|}{104.3} &
  \textbf{5} &
  \textbf{57} &
  \textbf{90} &
  \textbf{0.278} &
  4 &
  NA &
  NA &
  \multicolumn{1}{r|}{NA} &
  \textbf{6} &
  NA &
  NA &
  NA \\
tech-RL-caida &
  190914 &
  607610 &
  17 &
  112 &
  1788 &
  \multicolumn{1}{r|}{0.973} &
  \textbf{18} &
  \textbf{58} &
  \textbf{1016} &
  \textbf{0.735} &
  17 &
  636 &
  9037 &
  \multicolumn{1}{r|}{12.8} &
  \textbf{21} &
  \textbf{203} &
  \textbf{3371} &
  \textbf{4.34} \\
tech-routers-rf &
  2113 &
  6632 &
  16 &
  30 &
  321 &
  \multicolumn{1}{r|}{0.001} &
  \textbf{17} &
  \textbf{21} &
  \textbf{194} &
  \textbf{0.000} &
  16 &
  128 &
  1097 &
  \multicolumn{1}{r|}{0.014} &
  \textbf{18} &
  \textbf{91} &
  \textbf{826} &
  \textbf{0.009} \\
tech-WHOIS &
  7476 &
  56943 &
  47 &
  255 &
  15277 &
  \multicolumn{1}{r|}{\textbf{48.95}} &
  \textbf{48} &
  \textbf{250} &
  \textbf{14918} &
  52.90 &
  47 &
  278 &
  16755 &
  \multicolumn{1}{r|}{95.99} &
  \textbf{51} &
  \textbf{262} &
  \textbf{15751} &
  \textbf{79.06} \\
web-google &
  1299 &
  2773 &
  18 &
  19 &
  170 &
  \multicolumn{1}{r|}{0.000} &
  \textbf{19} &
  \textbf{0} &
  \textbf{0} &
  \textbf{0.000} &
  18 &
  115 &
  796 &
  \multicolumn{1}{r|}{\textbf{0.002}} &
  \textbf{19} &
  \textbf{93} &
  \textbf{686} &
  0.003 \\
web-polblogs &
  643 &
  2280 &
  9 &
  30 &
  253 &
  \multicolumn{1}{r|}{0.003} &
  \textbf{10} &
  \textbf{23} &
  \textbf{177} &
  \textbf{0.002} &
  9 &
  643 &
  2280 &
  \multicolumn{1}{r|}{0.564} &
  \textbf{12} &
  \textbf{122} &
  \textbf{768} &
  \textbf{0.031} \\
web-sk-2005 &
  121422 &
  334419 &
  82 &
  248 &
  10123 &
  \multicolumn{1}{r|}{0.071} &
  \textbf{83} &
  \textbf{246} &
  \textbf{9963} &
  \textbf{0.013} &
  82 &
  250 &
  10279 &
  \multicolumn{1}{r|}{0.096} &
  \textbf{83} &
  \textbf{250} &
  \textbf{10279} &
  \textbf{0.042} \\
web-spam &
  4767 &
  37375 &
  20 &
  108 &
  2282 &
  \multicolumn{1}{r|}{0.167} &
  \textbf{21} &
  \textbf{91} &
  \textbf{1783} &
  \textbf{0.158} &
  20 &
  220 &
  4805 &
  \multicolumn{1}{r|}{7.249} &
  \textbf{21} &
  \textbf{198} &
  \textbf{4350} &
  \textbf{7.205}
\\\bottomrule
\end{tabular}
}
}
\vspace{-0.5em}\caption{Detailed results of KDBB and KDBB$_{LB}^+$ in RealWorld benchmark with $s=2$ and $s=8$. Better results appear in bold.}
\label{tab-appendix-Real}
\end{table*}
\begin{table*}[!ht]
\centering
\footnotesize
\scalebox{1}{
\resizebox{\linewidth}{!}{
\begin{tabular}{l|rr|rrrrrrrr|rrrrrrrr}\toprule
\multirow{3}{*}{Instance} &
  \multicolumn{1}{c}{\multirow{3}{*}{$|V|$}} &
  \multicolumn{1}{c|}{\multirow{3}{*}{$|E|$}} &
  \multicolumn{8}{c|}{$s=2$} &
  \multicolumn{8}{c}{$s=8$} \\ \cline{4-19} 
 &
  \multicolumn{1}{c}{} &
  \multicolumn{1}{c|}{} &
  \multicolumn{4}{c|}{KDBB} &
  \multicolumn{4}{c|}{KDBB$_{LB}^+$} &
  \multicolumn{4}{c|}{KDBB} &
  \multicolumn{4}{c}{KDBB$_{LB}^+$} \\
 &
  \multicolumn{1}{c}{} &
  \multicolumn{1}{c|}{} &
  \multicolumn{1}{r}{LB} &
  \multicolumn{1}{r}{$|V'|$} &
  \multicolumn{1}{r}{$|E'|$} &
  \multicolumn{1}{r|}{Time} &
  \multicolumn{1}{r}{LB} &
  \multicolumn{1}{r}{$|V'|$} &
  \multicolumn{1}{r}{$|E'|$} &
  Time &
  \multicolumn{1}{r}{LB} &
  \multicolumn{1}{r}{$|V'|$} &
  \multicolumn{1}{r}{$|E'|$} &
  \multicolumn{1}{r|}{Time} &
  \multicolumn{1}{r}{LB} &
  \multicolumn{1}{r}{$|V'|$} &
  \multicolumn{1}{r}{$|E'|$} &
  Time \\ \hline
brock200-2.col &
  200 &
  9876 &
  \multicolumn{1}{r}{10} &
  \multicolumn{1}{r}{200} &
  \multicolumn{1}{r}{9876} &
  \multicolumn{1}{r|}{\textbf{1206}} &
  \multicolumn{1}{r}{\textbf{11}} &
  \multicolumn{1}{r}{200} &
  \multicolumn{1}{r}{9876} &
  1245 &
  \multicolumn{1}{r}{22} &
  \multicolumn{1}{r}{NA} &
  \multicolumn{1}{r}{NA} &
  \multicolumn{1}{r|}{NA} &
  \multicolumn{1}{r}{\textbf{25}} &
  \multicolumn{1}{r}{NA} &
  \multicolumn{1}{r}{NA} &
  NA \\
c-fat200-1.col &
  200 &
  1534 &
  12 &
  100 &
  809 &
  \multicolumn{1}{r|}{\textbf{0.004}} &
  12 &
  100 &
  809 &
  0.005 &
  8 &
  200 &
  1534 &
  \multicolumn{1}{r|}{0.038} &
  \textbf{11} &
  200 &
  1534 &
  \textbf{0.032} \\
c-fat200-2.col &
  200 &
  3235 &
  24 &
  46 &
  650 &
  \multicolumn{1}{r|}{\textbf{0.002}} &
  24 &
  46 &
  650 &
  0.003 &
  34 &
  200 &
  3235 &
  \multicolumn{1}{r|}{0.042} &
  \textbf{38} &
  200 &
  3235 &
  \textbf{0.039} \\
c-fat200-5.col &
  200 &
  8473 &
  58 &
  172 &
  6527 &
  \multicolumn{1}{r|}{\textbf{0.018}} &
  58 &
  172 &
  6527 &
  0.020 &
  22 &
  200 &
  8473 &
  \multicolumn{1}{r|}{\textbf{0.057}} &
  \textbf{26} &
  200 &
  8473 &
  0.059 \\
c-fat500-1.col &
  500 &
  4459 &
  14 &
  152 &
  1465 &
  \multicolumn{1}{r|}{0.012} &
  14 &
  152 &
  1465 &
  \textbf{0.010} &
  46 &
  500 &
  4459 &
  \multicolumn{1}{r|}{0.370} &
  \textbf{50} &
  500 &
  4459 &
  \textbf{0.337} \\
c-fat500-10.col &
  500 &
  46627 &
  126 &
  376 &
  31313 &
  \multicolumn{1}{r|}{0.213} &
  126 &
  376 &
  31313 &
  \textbf{0.191} &
  22 &
  500 &
  46627 &
  \multicolumn{1}{r|}{0.779} &
  \textbf{26} &
  500 &
  46627 &
  \textbf{0.745} \\
c-fat500-2.col &
  500 &
  9139 &
  26 &
  284 &
  5215 &
  \multicolumn{1}{r|}{0.034} &
  26 &
  284 &
  5215 &
  \textbf{0.033} &
  9 &
  500 &
  9139 &
  \multicolumn{1}{r|}{\textbf{0.396}} &
  \textbf{12} &
  500 &
  9139 &
  0.398 \\
c-fat500-5.col &
  500 &
  23191 &
  64 &
  190 &
  7970 &
  \multicolumn{1}{r|}{0.051} &
  64 &
  190 &
  7970 &
  \textbf{0.044} &
  22 &
  500 &
  23191 &
  \multicolumn{1}{r|}{\textbf{0.500}} &
  \textbf{26} &
  500 &
  23191 &
  0.503 \\
hamming6-2.col &
  64 &
  1824 &
  32 &
  64 &
  1824 &
  \multicolumn{1}{r|}{\textbf{1.150}} &
  32 &
  64 &
  1824 &
  1.202 &
  18 &
  64 &
  1824 &
  \multicolumn{1}{r|}{\textbf{10.67}} &
  \textbf{22} &
  64 &
  1824 &
  13.65 \\
hamming6-4.col &
  64 &
  704 &
  4 &
  64 &
  480 &
  \multicolumn{1}{r|}{0.026} &
  \textbf{5} &
  64 &
  480 &
  \textbf{0.022} &
  59 &
  64 &
  704 &
  \multicolumn{1}{r|}{\textbf{1.333}} &
  \textbf{63} &
  64 &
  704 &
  1.580 \\
johnson8-2-4.col &
  28 &
  210 &
  4 &
  28 &
  210 &
  \multicolumn{1}{r|}{0.016} &
  \textbf{5} &
  28 &
  210 &
  \textbf{0.014} &
  35 &
  28 &
  210 &
  \multicolumn{1}{r|}{\textbf{0.051}} &
  \textbf{39} &
  28 &
  210 &
  0.055 \\
johnson8-4-4.col &
  70 &
  1855 &
  14 &
  70 &
  1855 &
  \multicolumn{1}{r|}{\textbf{12.64}} &
  14 &
  70 &
  1855 &
  13.22 &
  12 &
  70 &
  1855 &
  \multicolumn{1}{r|}{\textbf{374.7}} &
  \textbf{15} &
  70 &
  1855 &
  388.4 \\
MANN-a9.col &
  45 &
  918 &
  16 &
  45 &
  918 &
  \multicolumn{1}{r|}{\textbf{20.16}} &
  \textbf{17} &
  45 &
  918 &
  20.58 &
  14 &
  45 &
  918 &
  \multicolumn{1}{r|}{\textbf{63.41}} &
  \textbf{17} &
  45 &
  918 &
  64.61 \\
p-hat300-1.col &
  300 &
  10933 &
  8 &
  300 &
  10894 &
  \multicolumn{1}{r|}{\textbf{440.3}} &
  \textbf{9} &
  300 &
  \textbf{10842} &
  440.8 &
  \multicolumn{1}{r}{16} &
  \multicolumn{1}{r}{NA} &
  \multicolumn{1}{r}{NA} &
  \multicolumn{1}{r|}{NA} &
  \multicolumn{1}{r}{\textbf{20}} &
  \multicolumn{1}{r}{NA} &
  \multicolumn{1}{r}{NA} &
  NA
\\\bottomrule
\end{tabular}
}
}
\vspace{-0.5em}\caption{Detailed results of KDBB and KDBB$_{LB}^+$ in DIAMCS2 benchmark with $s=2$ and $s=8$. Better results appear in bold.\vspace{-1em}}
\label{tab-appendix-DIAMCS2}
\end{table*}


In Section 5.5 of the main text, we show that our method of generating the initial lower bound (LB) can also be used to boost the initial lower bound of another relaxation clique problem, the Maximum $s$-Defective Clique Problem (MDCP)~\cite{k-defective-pro,KDBB}. 
Here we further apply our lower bound method based on short random walks~\cite{pearson1905problem} to typical exact algorithms for MDCP, to show 
that our new lower bound could actually boost the algorithm performance. 
The KDBB algorithm~\cite{KDBB} is the state-of-the-art BnB algorithm for MDCP, which uses a greedy construction heuristic guided by the vertex degree to generate the initial LB. We replace the construction heuristic in KDBB with our short random walks-based heuristic and denote the resulting algorithm as KDBB$_{LB}^+$.

We compare KDBB$_{LB}^+$ with KDBB in the RealWorld\footnote{http://lcs.ios.ac.cn/\%7Ecaisw/Resource/realworld\%20\\graphs.tar.gz} and DIMACS2\footnote{http://archive.dimacs.rutgers.edu/pub/challenge/graph/\\benchmarks/clique/} benchmarks, which contain massive sparse and dense graphs, respectively, with $s = 2$ and $s = 8$. The cut-off time is set to 3,600 seconds. The detailed comparison results are summarized in Tables~\ref{tab-appendix-Real} and~\ref{tab-appendix-DIAMCS2}, where columns $V$ and $E$ indicate the number of vertices and edges in the original graph, respectively, columns $V'$ and $E'$ indicate the number of vertices and edges in the reduced graph after the preprocessing stage in the algorithm, respectively, column \textit{LB} indicates the initial LB obtained by the algorithm, and column \textit{Time} indicates the running time in seconds to solve the instance. The symbol `NA' indicates that the algorithm cannot solve the instance within the given cut-off time. Note that our method has never made the initial LB worse in the tested instances. In the tables, we only provide the results of instances solved by at least one algorithm within the cut-off time, and at the same time KDBB$_{LB}^+$ obtains a larger initial LB than KDBB, when $s = 2$ or $s = 8$.



Experimental results indicate that, for the DIMACS2 benchmark shown in Table~\ref{tab-appendix-DIAMCS2}, although the initial LB in KDBB$_{LB}^+$ has exhibited enhancements, especially when $s=8$, it is hard to reduce vertices and edges in dense graphs. Therefore, the overall graph size remains largely unaffected in the majority of computational scenarios, and the running time for solving the instances of the two algorithms is close. 
Conversely, for the RealWorld benchmark detailed in Table~\ref{tab-appendix-Real}, the improvements in the initial LB by our method demonstrate a notable reduction in graph size, consequently enhancing the algorithm efficiency in most instances. For example, when $s=2$, the algorithm's efficiency experiences a substantial boost owing to a significant reduction in ca-CSphd. Furthermore, when $s=8$, KDBB fails to solve ia-wiki-Talk within the cut-off time, whereas KDBB$_{LB}^+$ can solve it successfully. The results demonstrate again the excellent performance and generalization capability of our LB method, especially on sparse graphs which is prevalent in the real world.

\end{document}